\documentclass[12pt]{article}

\usepackage{amsfonts}
\usepackage{amsmath}
\usepackage{amssymb}
\usepackage{amsthm}
\usepackage{mathtools}
\usepackage{authblk}
\usepackage{appendix}
\usepackage{booktabs}
\usepackage{bm}
\usepackage{bbm}
\usepackage{dsfont}
\usepackage[dvipsnames]{xcolor}
\usepackage{natbib}
\usepackage{subcaption}
\usepackage{lineno}
\usepackage[hidelinks]{hyperref}
\usepackage{enumitem}
\usepackage[T1]{fontenc}
\usepackage{multirow}
\usepackage{soul}
\usepackage[font = small]{caption}
\usepackage{cancel}
\usepackage{algorithmicx}
\usepackage{algorithm}
\usepackage{algpseudocode}

\usepackage{graphicx,url}

\newtheorem{theorem}{Theorem}
\newtheorem{lemma}{Lemma}

\newtheorem{corollary}{Corollary}
\newtheorem{assumption}{Assumption}

\newtheorem{assumptionalt}{}[assumption]

\newenvironment{assumptionp}[1]{
  \renewcommand\theassumptionalt{#1}
  \assumptionalt
}{\endassumptionalt}

\usepackage[utf8]{inputenc}  

\graphicspath{ {../images/} }

\newcommand{\blind}{1}
\newcommand{\spacing}{1.1}

\newcommand{\norm}[1]{\left\lVert#1\right\rVert}

\newcommand{\bA}{\textbf{A}}
\newcommand{\bL}{\textbf{L}}
\newcommand{\bP}{\textbf{P}}
\newcommand{\bQ}{\textbf{Q}}

\newcommand{\bR}{\textbf{R}}
\newcommand{\bI}{\textbf{I}}
\newcommand{\bY}{\textbf{Y}}
\newcommand{\bx}{\textbf{x}}
\newcommand{\rd}{\mathrm{d}}
\newcommand{\pp}{\mathbb{P}}
\newcommand{\Dir}{\mathrm{Dir}}
\newcommand{\balpha}{\bm{\alpha}}
\newcommand{\bphi}{\bm{\phi}}
\newcommand{\bpsi}{\bm{\psi}}
\newcommand{\bdelta}{\bm{\delta}}
\newcommand{\T}{\intercal}
\newcommand{\one}{\mathbbm{1}}
\newcommand{\E}{\mathbb{E}}

\begin{document}

\def\spacingset#1{\renewcommand{\baselinestretch}%
{#1}\small\normalsize} \spacingset{1}

\if1\blind
{
  \title{\bf Efficient Bayesian Inference for Discretely Observed Continuous Time Markov Chains}
    \author[1]{Tao Tang}
    \author[2]{Lachlan Astfalck}
    \author[3]{David Dunson}
    \affil[1]{Department of Mathematics, Duke University, Durham, NC, USA}
    \affil[2]{School of Physics, Mathematics \& Computing, The University of Western Australia, Australia}
    \affil[3]{Department of Statistical Science, Duke University, Durham, NC, USA}
    
    \setcounter{Maxaffil}{0}
    \renewcommand\Affilfont{\itshape\small}
  \maketitle
} \fi

\if0\blind
{
  \bigskip
  \bigskip
  \bigskip
  \begin{center}
    {\LARGE\bf Low-Rank Estimation of Continuous Time Markov Chains}
\end{center}
  \medskip
} \fi

\bigskip
\begin{abstract}
  Inference for continuous-time Markov chains (CTMCs) becomes challenging when the process is only observed at discrete time points. The exact likelihood is intractable, and existing methods often struggle even in medium-dimensional state-spaces. We propose a scalable Bayesian framework for CTMC inference based on a pseudo-likelihood that bypasses the need for the full intractable likelihood. Our approach jointly estimates the probability transition matrix and a biorthogonal spectral decomposition of the generator, enabling an efficient Gibbs sampling procedure that obeys embeddability. Existing methods typically integrate out the unobserved transitions, which becomes computationally burdensome as the number of data or dimensions increase. The computational cost of our method is near-invariant in the number of data and scales well to medium-high dimensions. We justify our pseudo-likelihood approach by establishing  theoretical guarantees, including a Bernstein–von Mises theorem for the probability transition matrix and posterior consistency for the spectral parameters of the generator. Through simulation and applications, we showcase the flexibility and robustness of our approach, offering a tractable and scalable approach to Bayesian inference for CTMCs.
\end{abstract}

\noindent%

\spacingset{\spacing}

\section{Introduction} \label{sec:intro}

Estimating the dynamics of a continuous-time Markov chain (CTMC) from discretely observed states remains challenging, especially for large state spaces or long time series. The generator matrix $\bL$ that governs the CTMC must satisfy structural constraints, which complicate likelihood evaluation and render full Bayesian inference computationally demanding and often ill-posed. When, instead, the sequence of states and the exact holding times are observed, the CTMC is fully observed and the likelihood is tractable in $\bL$. However, such data are rarely available in application, where typically only the states at discrete time points are recorded. Examples include biological systems \citep{zhao2016bayesian}, healthcare \citep{tahami2023estimating}, finance \citep{inamura2006estimating,dos2020capturing}, physical and environmental processes \citep{lienard2016modelling}, genetics \citep{hajiaghayi2014efficient}, and queueing systems \citep{asmussen2003applied}; all of which sample observations discretely. In principle, Bayesian estimation can accommodate such data by treating the unobserved holding times as missing and integrating over them during inference \citep[e.g.][]{bladt2005statistical}. In practice, however, this becomes computationally intensive as the size of the state space and the number of observations grow. In addition to computational barriers, a more fundamental problem arises from the fact that the discrete-time transitions observed in the data may be consistent with many different generators. This not only induces multi-modality in the posterior over $\bL$, but also undermines the identifiability of the underlying dynamics, potentially leading to dramatically different conclusions when extrapolating beyond the sampling interval. We present a scalable Bayesian method, for CTMCs with medium-sized (10s–100s) states, that avoids latent path integration and resolves non-identifiability through a biorthogonal parametrization of $\bL$, thus ensuring structural validity and identifiability.

Several approaches have been proposed for inference of discretely observed CTMCs. In Bayesian inference, the most common strategy is to treat the unobserved holding times as missing data and sample full latent paths, after which inference is straightforward via standard Gibbs updates \citep{bladt2005statistical,fearnhead2006exact}. Simple Metropolis-Hastings algorithms that do not rely on latent path integration have been proposed \citep{riva2023estimation}, although they require delicate tuning of the proposal densities.
Both approaches are mathematically valid; however, computational issues arise for large data or dimensions. Methods have been developed for high/infinite-dimensional state spaces, although these methods rely on lower-dimensional or structured parametrizations of the generator that implicitly enforce the required constraints, and do not estimate $\bL$ directly from the data. Examples include \cite{sherlock2021direct} who evaluate the likelihood directly using computationally efficient properties of the matrix exponential; \cite{biron2023pseudo} who develop a pseudo-marginal likelihood approximation; and \cite{hajiaghayi2014efficient} who estimate the generator using a particle-based Monte Carlo approach. Outside of Bayesian methods, optimization approaches exist \cite[e.g.][]{israel2001finding,metzner2007generator} although typically these methods lack principled uncertainty quantification and struggle in higher-dimensional state spaces. Notably, \cite{crommelin2006fitting} also constrain $\bL$ to the space of biorthogonal matrices. However, as we discuss in detail below, their method imposes low-rank constraints on $\bL$ that render the chain reducible, limiting the practical applicability.

Recent research has primarily focused on the computational aspects of algorithmic development; however, outside of \cite{crommelin2006fitting}, limited attention has been given to the embedding problem for Markov chains \citep[see][for rigorous definition]{kingman1962imbedding}. When the CTMC is observed at regular intervals, the likelihood depends on the generator $\bL$ only through the corresponding discrete-time probability transition matrix $\bP = \exp\{\Delta \bL\}$, where $\Delta$ denotes the sampling interval and, here, $\exp\{\cdot\} = \sum_{k=0}^\infty (\cdot)^k/k!$ is the matrix exponential. In principle, we may estimate $\bP$ directly and then recover $\bL$ via the matrix logarithm. However, the mapping from $\bP$ to $\bL$ is one-to-many, and further, not every transition matrix $\bP$ admits a valid generator. This implies that even when $\bP$ is estimated accurately, recovering a valid $\bL$ may be impossible, or may yield multiple incompatible solutions. In general, the embedding problem renders inference on $\bL$ as fundamentally ill-posed. An analytical solution is available for $2 \times 2$ matrices \citep[see][]{kingman1962imbedding} and $3 \times 3$ matrices \citep[see][]{carette1995characterizations}, but limited findings are available for the general case of $n \times n$ matrices. 
To make inference tractable and well-posed, we adopt a parametrization that directly encodes the generator $\bL$ in a way that guarantees both structural validity and identifiability from discrete-time data.

Although the existing literature offers valuable contributions, important limitations remain. Optimization-based approaches typically yield only point estimates and lack principled uncertainty quantification. Bayesian methods typically treat generator estimation as a missing-data problem and are computationally expensive and difficult to scale to large data and state spaces. Methods such as \cite{sherlock2021direct} designed for large state-spaces are predicated on some known parametrization of $\bL$ and do not respect embeddability if estimating $\bL$ directly. None of the existing approaches combine scalability, coherent uncertainty quantification, and estimation of an embeddable $\bL$ in a unified framework. To address this, we propose a pseudo-Bayesian approach built upon a novel pseudo-likelihood. We restrict $\bL$ to the class of bi-orthogonal matrices parametrized by its spectral decomposition $\{\lambda_k, \phi_k, \psi_k\}_{k=1}^m$, ensuring uniqueness through the matrix exponential and logarithm. Rather than estimating $\bL$ directly from the data, the pseudo-likelihood jointly infers the transition matrix $\bP$ and the spectral parameters $\{\lambda_k, \phi_k, \psi_k\}_{k=1}^m$. In effect, the pseudo-likelihood combines the empirical transition structure of the observed data with a regularization term that favors agreement of $\bP$ with $\bL$. By estimating $\bP$ and $\bL$ as parameters with a joint probability structure, our method circumvents the need to explicitly enforce embeddability in $\bP$, enabling efficient and scalable inference.

We provide rigorous theoretical justification for our pseudo-Bayesian approach with two main contributions. First, we establish a Bernstein--von Mises (BvM) theorem for the posterior distribution of $\bP$, showing that it concentrates around the empirical estimator at the standard parametric rate. The proof requires non-standard tools to account for the dependence structure induced by the Markovian sampling process. Second, we prove posterior consistency for the spectral parameters under a generative CTMC model. Together, these results imply that the empirical transition matrix converges to the true $\bP$, and the spectral decomposition of $\bL$ converges under the true model. This provides a theoretical foundation for the suitability of our pseudo-likelihood approach. We empirically support this with simulation studies where we compare our method against existing literature and demonstrate equivalent accuracy with significantly better computational performance. Finally, we examine a larger-dimensional state space example and provide an application to a meta-stable diffusion process canonical in the literature.

This article proceeds as follows. Section~\ref{sec:ctmc} provides a mathematical overview of CTMCs, introduces the model and its spectral decomposition, and defines the proposed pseudo-likelihood. Section~\ref{sec:inference} describes the pseudo-posterior inference procedure. Section~\ref{sec:justification} establishes the theoretical guarantees of our method, including a BvM theorem and posterior consistency. Sections~\ref{sec:inference} and \ref{sec:justification} are modular and may be read independently, depending on the reader's interest. Section~\ref{sec:applications} presents numerical experiments and applications, and Section~\ref{sec:conclusion} concludes.

\section{Modeling Continuous Time Markov Chains} \label{sec:ctmc}

\subsection{Definition}

A CTMC is a continuous-time stochastic process $\{X_t : t \geq 0\}$ defined in the probability space $(\Omega, \mathcal{F}, \mathbb{P})$. For any time $t \geq 0$, the CTMC takes values in the countable set $\mathcal{S}$. Equipped with the discrete metric and constructed via jump processes, any sample path $t \mapsto X_t(\omega)$ for $\omega \in \Omega$ is right continuous with left limits. We formally define a CTMC by three elements: (1) the state space $\mathcal{S}$, (2) a transition rate matrix, or generator, $\textbf{L}$ with dimension $|\mathcal{S}| \times |\mathcal{S}|$ that describes the dynamics of the process, and (3) the distribution of the initial state $X_0 \sim \pi_0$. We assume time-homogeneity of the process and $\mathcal{S}$ to be discrete so that $\mathcal{S} = \{1, \dots, m\}$, and thus $\textbf{L}$ is an $m \times m$ time-invariant matrix. Denote by $\bL(p,q)$ the $(p,q)$th element of $\bL$; for $p \neq q$ the element $\bL(p,q)$ is non-negative and describes the instantaneous rate of transition from state $p$ to state $q$. For $p = q$, $\bL(p,p)$ is the instantaneous rate of leaving state $p$, defined as $\bL(p,p) = - \sum_{q \neq p} \bL(p,q)$, ensuring the rows of $\bL$ sum to zero for all $p$. Herein, we assume that the full continuous sample path is not observed; instead, the process is observed only at a finite collection of discrete time points.

From $\bL$, we can define the transition probability matrix $\bP_t$ for $t \geq 0$ with $(p,q)$th elements $\bP_t(p, q) = \pp(X_t = q \mid X_0 = p)$. The transition probability matrix, $\bP_t$, satisfies the Kolmogorov forward equation
\begin{equation*} \label{eqn:kolmogorov}
  \frac{\rd \bP_t}{\rd t} = \bP_t \bL, \quad \bP_0 = \bI,
\end{equation*}
with solution $\bP_t = \mathrm{exp}\{t \bL\} = \sum_{n=0}^{\infty} \frac{(t \bL)^n}{n!}$, where here $\mathrm{exp}\{\cdot\}$ denotes the matrix exponential. Thus, in theory, if we can obtain $\bP$ so too do we obtain $\bL$ via $\bL = t^{-1} \log \{\bP_t\}$; although, in practice, there are a number of difficulties.

First, computation of the matrix exponential and logarithm is often intractable, and so even if $\bP_t$ is known exactly, not always will $\bL$ be computationally available. A sufficient condition that guarantees computational tractability is when $\bL$ admits a biorthogonal decomposition $\bL = \Phi \Lambda \Psi^\T$ such that $\Lambda$ is diagonal and $\Phi^\T \Psi = \bI$. Then, $\exp\{\bL\} = \Phi \exp\{\Lambda\} \Psi^\T$ and $\exp\{\Lambda\}_{ii} = \exp\{\Lambda_{ii}\}$. Even still, in practice $\bP_t$ is unknown and so must be estimated by some $\hat{\bP}_t$. When the data are observed irregularly, $\bP_t$ must be estimated as a function of continuous $t$ and this may become prohibitively difficult in practice. Even for regularly observed data, for instance, given a length $n$ observation with regular observations $\{x_{j}\}$, for ${j \in \{1, \dots, n\}}$, the estimator
\begin{equation} \label{eqn:P_tilde}
  \hat{\bP}_t = \frac{\sum_{i = 0}^n \mathbbm{1}\{x_{it} = p\} \mathbbm{1}\{x_{(i+1)t} = q\}}{\sum_{i = 0}^n \mathbbm{1}\{x_{it} = p\}},
\end{equation}
where $\mathbbm{1}\{\cdot\}$ is the indicator function, is a valid transition probability matrix such that $\sum_{q} \hat{\bP}_t(p,q) = 1$ for all $q$ and $\hat{\bP}_t(p, q) \geq 0$ for all $p,q$. However, $\hat{\bL} = t^{-1} \log \{\hat{\bP}_t\}$ will not necessarily satisfy embeddability and so is not necessarily a valid generator of a CTMC.

Our goal is to infer $\bL$, with uncertainty, that characterizes the CTMC. 
Estimation of $\bL$ in continuous time is more challenging than in discrete time, as the likelihood is generally intractable. This issue arises as in continuous time we do not observe the transition times but only have estimation of the state at discrete time-steps; we describe this in more detail in Section~\ref{sec:spec_decomp}, below. In this paper, we offer a solution for when $\bL$ is biorthogonal and the data $\{x_t\}_{t = 1}^n$ are observed at equally spaced times $t_i = i \Delta$.

\subsection{Spectral decomposition of the generator} \label{sec:spec_decomp}

Similar to \cite{crommelin2006fitting}, we define the generator matrix $\bL$ via its spectral decomposition,
\begin{equation} \label{eqn:L_spectral}
  \bL(p,q) = \sum_{k=1}^m \lambda_k \phi_k(p) \psi_k(q),
\end{equation}
where $\lambda_k$ are the eigenvalues of $\bL$, and $\phi_k, \psi_k \in \mathbb{R}^m$ are the corresponding right and left eigenvectors that define the $k$th columns of $\Phi$ and $\Psi$, respectively. These satisfy
\begin{equation*}
  \bL \phi_k = \lambda_k \phi_k, \quad \psi_k^\T \bL = \lambda_k \psi_k^\T, \quad \text{and} \quad \psi_k^\T \phi_l = \delta_{kl},
\end{equation*}
where $\delta_{kl} \coloneqq \mathbbm{1}\{k = l\}$ denotes the Kronecker delta. The eigenvectors are ordered so that $\mathrm{Re}(\lambda_k) \geq \mathrm{Re}(\lambda_{k+1})$, and we denote the $p$th components of the eigenvectors as $\phi_k(p)$ and $\psi_k(q)$ for $p, q \in \mathcal{S}$. For the representation in \eqref{eqn:L_spectral} to provide a valid probability matrix, we have $\lambda_1 = 0$, $\phi_1 = (1, \dots, 1)^\T$, and the corresponding left eigenvector $\psi_1$ encodes the invariant distribution, $\psi_1^\T = \mu$.

The spectral representation of $\bL$ in \eqref{eqn:L_spectral} provides a natural avenue through which we can compute the matrix exponential. Indeed, the transition probability matrix has the representation
\begin{equation*}
  \bP_\Delta = \sum_{k=1}^m \exp\{\lambda_k \Delta\} \phi_k \psi_k^\T,
\end{equation*}
where $\Delta$ denotes the sampling interval. Where unambiguous, we suppress the dependence on $\Delta$ by writing $\bP_\Delta \equiv \bP$.

If the full CTMC path were observed, including the sequence of visited states $\mathbf{x}_r = \{x_0, x_1, \dots, x_r\}$ and holding times $\mathbf{t}_r = \{t_0, t_1, \dots, t_r\}$, then the likelihood is tractable in $\bL$ and takes the form
\begin{equation} \label{eqn:L_likelihood}
  \mathcal{L}(\mathbf{x}_r, \mathbf{t}_r \mid \bL) = \prod_{p \neq q} \bL(p,q)^{N_{pq}} \prod_p \exp\{\bL(p,p) T_p\},
\end{equation}
where $N_{pq} = \sum_{i = 1}^r \mathbbm{1}\{x_{i-1} = p, x_i = q\}$ counts the number of transitions from $p$ to $q$, and $T_p = \sum_{i = 0}^r t_i \mathbbm{1}\{x_i = p\}$ is the total time spent in state $p$. 

In contrast, for discrete-time observations $\mathbf{x}_n = \{x_0, x_1, \dots, x_n\}$ sampled at interval $\Delta$, the discrete-time likelihood corresponds to that of a discrete-time Markov chain,
\begin{equation} \label{eqn:P_likelihood}
  \mathcal{L}(\mathbf{x}_n \mid \bP) = \prod_{i=1}^n \bP(x_{i-1}, x_i).
\end{equation}
Although the complete-data likelihood \eqref{eqn:L_likelihood} is tractable in $\bL$, the discretely observed likelihood \eqref{eqn:P_likelihood} depends on $\bL$ through the matrix exponential, making it a non-linear and computationally intensive function of $\bL$.

Assuming the spectral representation \eqref{eqn:L_spectral} exists, the likelihood can be re-expressed in terms of the spectral components,
\begin{equation} \label{eqn:spec_likelihood}
  \mathcal{L}(\mathbf{x}_n \mid \bP) = \mathcal{L}(\mathbf{x}_n \mid \{\lambda_k, \phi_k, \psi_k\}_{k=1}^m) = \prod_{i=1}^n \left\{\sum_{k=1}^m \exp\{\lambda_k \Delta\} \phi_k(x_{i-1}) \psi_k(x_i)\right\}.
\end{equation}
This parametrization provides an explicit link between the probability transition structure and the generator. However, it still requires knowledge of the eigen-decomposition of $\bL$, which is generally unknown in practice.

\subsection{Generator Rank}
In the discrete-time Markov chain literature, a common simplifying assumption is to assume low-rank structure in $\bP$ so that
\[\bP \approx \sum_{k=1}^d 
\Lambda_k \phi_k \psi_k^\T,\]
where $d \ll m$ and $\Lambda_k$ is the $k$th eigenvalue of $\bP$ \citep[e.g.][]{zhang2019spectral}. \cite{crommelin2006fitting} extend this notion to the continuous-time case by similarly assuming a low-rank structure for $\bL$. However, assumptions on the rank of $\bL$ have direct consequences on the degree of reducibility and may lead to unintended and unrealistic behavior. Assume $\bL$, and so too $\bP$, describe an irreducible Markov chain. Define $\Lambda_k = \exp\{\lambda_k \Delta\}$, noting that this relationship holds always, and is not reliant on the assumption of biorthogonality. According to the Perron–Frobenius theorem, the spectral radius of an irreducible $\bP$ is 1, that is, the leading eigenvalue is $\Lambda_1 = 1$. Further, if $\bP$ is also aperiodic then $|\Lambda_i| < 1$, $i \geq 2$ with strict inequality. Thus, $\lambda_1 = 0$ as we have previously noted, and $\lambda_i < 0$ for $i \geq 2$. As there is exactly one zero eigenvalue and the rest are guaranteed to have strictly negative real components, $\mathrm{rank}(\bL) = m-1.$ This result may be further generalized to say that if $c$ is the number of closed communicating classes, then $\mathrm{rank}(\bL) = m-c.$ 
Thus, imposing the rank of $\bL$ and setting $\mathrm{rank}(\bL) = d \ll m$ implicitly imposes $m-d$ closed communicating classes and, outside the trivial case of $d = m-1$, contradicts any assumption of irreducibility.

\subsection{Definition of the pseudo-likelihood} \label{sec:low_rank}

The generator $\bL$ is related to the data through $\bP = \exp\{\bL \Delta\}$. In principle, inference can proceed in either direction: we may first estimate $\bP$ from the data and compute $\bL$ via the matrix logarithm, or instead posit a model for $\bL$ and evaluate its fit via the implied $\bP$. However, both these approaches are deficient. As we have discussed, the set of embeddable matrices is a strict subset of all stochastic matrices and so not every stochastic matrix is embeddable. Thus, the matrix logarithm of an empirical $\bP$ may not yield a valid generator. Direct modeling of $\bL$ does not yield a tractable form for conditional updates under the discrete-time likelihood in~\eqref{eqn:P_likelihood}. We argue that a principled alternative is to model $\bP$ and $\bL$ jointly, treating them as coupled objects linked by the exponential map. In particular, we assume that $\bL$ admits a biorthogonal decomposition, and impose this structure by regularizing $\bP$ toward the spectral form $\sum_{k=1}^m \exp\{\lambda_k \Delta\} \phi_k \psi_k^\T$. This avoids the embedding problem whilst enabling tractable inference. Define the pseudo-likelihood,
\begin{equation} \label{eqn:pseudo_likelihood}
  \log \mathcal{L}\left(\mathbf{x}_n \mid \bP, \{\lambda_k, \phi_k, \psi_k\}_{k=1}^m\right) = \sum_{i=1}^n \log \bP(x_{i-1}, x_{i}) - \frac{\nu}{2} \lVert \bP - \sum_{k=1}^m \exp\{\lambda_k \Delta\} \phi_k \psi_k^\T \rVert_\mathrm{F}^2,
\end{equation}
where $\nu > 0$ is a regularization parameter and $\norm{\cdot}_\mathrm{F}$ denotes the Frobenius norm. This likelihood has two terms, a log-likelihood for the observed transitions under $\bP$, and a penalty that encourages agreement between $\bP$ and its spectral form.

To rigorously justify the use of \eqref{eqn:pseudo_likelihood} as a pseudo-likelihood, we establish theoretical guarantees in Section~\ref{sec:justification}, including a Bernstein--von Mises theorem for the posterior distribution of $\bP$ and posterior consistency for the spectral parameters $\{\lambda_k, \phi_k, \psi_k\}_{k=1}^m$. Before returning to these results, we complete the model specification by describing the prior distributions and inference strategy.

\section{Pseudo-posterior inference} \label{sec:inference}

Specification of the Bayesian model begins with the prior distribution over the unknown parameters, $\pp(\bP, \{\lambda_k, \phi_k, \psi_k\}_{k=1}^m)$. The parameter space can be high-dimensional and is subject to stringent constraints: $\bP$ must be a valid transition probability matrix, and the collection $\{\lambda_k, \phi_k, \psi_k\}_{k=1}^m$ must constitute a valid spectral decomposition. We therefore seek priors that respect these constraints, permit efficient computation, and are sufficiently expressive to encode informative prior beliefs when available. Assuming prior independence, we factor the joint prior as $\pp(\bP, \{\lambda_k, \phi_k, \psi_k\}_{k=1}^m) = \pp(\bP) \cdot \pp(\{\lambda_k, \phi_k, \psi_k\}_{k=1}^m)$, and we adopt (approximately) conjugate forms to enable efficient Gibbs sampling. The following inference method is provided as an Algorithm in Appendix~\ref{sec:gibbs}.

\subsection{Sampling the transition probability matrix}

For each row of $\bP$, $\bP(p, \cdot) \in \mathbb{R}^m$, we specify a Dirichlet prior $\bP(p, \cdot) \sim \Dir(\balpha)$ where $\balpha = (\alpha_1, \dots, \alpha_m)$ is a length-$m$ vector of positive values $\alpha_i > 0$. There is no conjugate form for the full conditional distribution $\pp(\bP \mid \bx_n, \{\lambda_k, \phi_k, \psi_k\}_{k=1}^m)$ when considering the full pseudo-likelihood in \eqref{eqn:pseudo_likelihood}. However, as either $n \rightarrow \infty$ or $\nu \rightarrow 0$, the posterior conditional distribution $\pp(\bP \mid \bx_n, \{\lambda_k, \phi_k, \psi_k\}_{k=1}^m)$ becomes increasingly dominated by the first term in \eqref{eqn:pseudo_likelihood}. In particular, the conditional distribution of each row $\bP(p, \cdot)$ converges in total variation of the conditional posterior, as either $n \rightarrow \infty$ or $\nu \rightarrow 0$, to a Dirichlet random variable, independent of the spectral parameters, such that
\begin{equation} \label{sec:dirichlet_approx}
  \bP(p, \cdot) \mid \bx_n, \{\lambda_k, \phi_k, \psi_k\}_{k=1}^m \xrightarrow{\text{TV}} \Dir(\balpha + \textbf{c}_p) 
\end{equation} 
where $\textbf{c}_p = (c_p^1, \dots, c_p^m)$ and $c_p^q = \sum_{i=0}^{n-1} \mathbbm{1}\{x_i = p, x_{i+1} = q\}$ for $q \in \{1, \dots, m\}$. We are now faced with a choice. If we are willing to treat \eqref{sec:dirichlet_approx} as a sufficiently accurate approximation, we can directly specify
\begin{equation*}
  \pp(\bP(p, \cdot) \mid \bx_n, \{\lambda_k, \phi_k, \psi_k\}_{k=1}^m) = \pp(\bP(p, \cdot) \mid \bx_n) = \Dir(\balpha + \textbf{c}_p),
\end{equation*}
trading off spectral regularization for computational efficiency. Alternatively, we may use
\begin{equation*}
  q(\bP) = \Dir(\balpha + \textbf{c}_p).
\end{equation*}
as a proposal distribution in a Metropolis-Hastings scheme. In this case, the acceptance probability converges to $1$ as $n \to \infty$, providing a built-in diagnostic for whether the approximation in \eqref{sec:dirichlet_approx} is adequate in practice.

\subsection{Sampling the spectral decomposition of the generator}

Prior specifications for the spectral parameters are complicated by structural requirements of $\bP$ and $\bL$; namely, $\bP$ must be a non-negative matrix with rows that sum to one, and $\bL$ must have non-negative off-diagonal elements with rows that sum to zero. Additionally, under our assumption in \eqref{eqn:L_spectral} that both $\bP$ and $\bL$ are biorthogonal, we require the left and right eigenvectors to satisfy biorthogonality. Enforcing strict orthogonality is possible but computationally challenging \cite[see, for instance,][]{jauch2020random,jauch2021monte}. Instead, we follow \cite{duan2020bayesian} and \cite{matuk2022bayesian}, and relax the strict orthogonality constraint by instead specifying priors that shrink toward the space of orthogonal matrices, i.e., the Stiefel manifold. Recall the eigenvalue of $\bP$ associated with $\lambda_i$, $\Lambda_k = \exp\{\Delta \lambda_k\}$, and factorize $\pp(\{\Lambda_k, \phi_k, \psi_k\}_{k=1}^m) = \pp(\{\Lambda_k\}_{k=1}^m) \cdot \pp(\{\phi_k, \psi_k\}_{k=1}^m)$. We fix $\Lambda_1 = 1$ and $\phi_1 = (1, \dots, 1)^\T$, and assign a uniform prior to the remaining eigenvalues $\{\Lambda_k\}_{k=2}^m$, 
\begin{equation} \label{eqn:value_prior}
  \{\Lambda_k\}_{k=2}^m \sim \mathrm{Uniform}(\mathcal{A}), \quad \mathcal{A} = \{\{\Lambda_k\}_{k=2}^m \mid 1 = \Lambda_1 > \Lambda_2 \geq \cdots \geq \Lambda_m > 0\} \subset \mathbb{R}^{m-1}.
\end{equation}
Up to an additive constant, we specify the eigenvector prior as
\begin{equation} \label{eqn:vector_prior}
  \log \pp(\{\phi_k, \psi_k\}_{k=1}^m) \propto \underbrace{-\sum_{k=2}^m \frac{\norm{\phi_k}^2_2}{2\sigma_\phi^2} -\sum_{k=1}^m \frac{\norm{\psi_k}^2_2}{2\sigma_\psi^2}}_\mathrm{identifiability} - \underbrace{\sum_{j, k =1}^m \frac{(\delta_{jk} - \psi_j^\T \phi_k)^2}{2\sigma_c^2}}_\mathrm{biorthogonality},
\end{equation}
where $\norm{\cdot}_2$ denotes the Euclidean norm. The first bracketed term corresponds to mean-zero multivariate Gaussian priors on the eigenvectors and regularizes their scale, thus preventing identifiability issues that arise from the rescaling invariance between $\lambda_k$ and the $\psi_k$ and $\phi_k$. The second bracketed term shrinks the left and right eigenvectors toward biorthogonality. The biorthogonality constraint, along with the fixed values $\phi_1 = (1, \dots, 1)^\T$ and $\Lambda_1$, implies that the rows of $\bP$ and $\bL$ sum to one and zero, respectively. Neither prior specification in \eqref{eqn:value_prior} or \eqref{eqn:vector_prior} explicitly enforces non-negativity of the off-diagonal entries in $\bL$; instead, this structural constraint is imposed via truncation of the conditional distributions, as described below. This truncation may be interpreted as inducing a uniform measure over the subset of valid diagonalizable generators. Strictly speaking, it is not an explicit prior density defined over the parameter space but is induced implicitly via truncation of the full conditional distributions. 

Let $\Lambda_{-k}=\{\Lambda_j\}_{j\neq k}$. The same notation will be used for $\phi_{-k}$ and $\psi_{-k}$. Due to the form of \eqref{eqn:pseudo_likelihood}, the full conditional distributions of the spectral parameters are conditionally independent of the data $\bx_n$ given $\bP$. For $k \geq 2$, the log of the full conditional distribution of $\Lambda_k$ is given by
\begin{align*}
  \log \pp(\Lambda_k \mid \{\phi_k, \psi_k\}_{k=1}^m, \Lambda_{-k}, \bP) &\propto \frac{-\nu}{2}\lVert{\bP - \sum_{k=1}^m \Lambda_k \phi_k \psi_k^\T}\rVert_\mathrm{F}^2 - \log \pp(\{\Lambda_k\}_{k=1}^m), \\
  &= \frac{-\nu}{2}\lVert{\bY_k - \Lambda_k \phi_k \psi_k^\T}\rVert_\mathrm{F}^2 - \log \pp(\{\Lambda_k\}_{k=1}^m),
\end{align*}
where $\bY_k = \bP - \sum_{j \neq k} \Lambda_j \phi_j \psi_j^\T$ is the residual matrix after subtracting all but the $k$th component. This corresponds to the normal distribution $\mathcal{N}(\mu_{\Lambda_k}, \sigma^2_{\Lambda_k})$ truncated to the interval $\mathcal{A}_\Lambda = (\Lambda_{k+1}, \Lambda_{k-1})$ with parameters
\begin{equation*}
  \mu_{\Lambda_k} = \nu \sigma_{\Lambda_k}^2 \sum_{p, q = 1}^m \phi_k(p) \psi_k(q) \bY_k(p, q), \quad \sigma_{\Lambda_k}^2 = \frac{1}{\nu (\norm{\phi_k}_2^2 \cdot \norm{\psi_k}_2^2)}.
\end{equation*}
To ensure that the generator $\bL = \sum_{k=1}^m \lambda_k \phi_k \psi_k^\T$ with $\lambda_k = \Delta^{-1} \log \Lambda_k$, has non-negative off-diagonal entries, we require $\bL(i,j) = \bL_{-k}(i,j) + \lambda_k \phi_k(i) \psi_k(j) \geq 0$ for all $i \ne j$, where $\bL_{-k} = \sum_{l \neq k} \lambda_l \phi_l \psi_l^\T$. Solving for $\lambda_k$ and exponentiating yields the additional truncation set
\begin{equation*}
  \Lambda_k \in \mathcal{A}_k = \left[ \exp\left\{\underset{\phi_k(p) \psi_k(q) > 0}{\mathrm{max}} \frac{\bL_{-k}(p,q)}{\phi_k(p) \psi_k(q)}\right\}, \exp\left\{\underset{\phi_k(p) \psi_k(q) < 0}{\mathrm{min}} \frac{\bL_{-k}(p,q)}{\phi_k(p) \psi_k(q)} \right\}\right],
\end{equation*}
for all $i \ne j$. Thus, the full conditional distribution is then given by
\begin{equation*}
  \pp(\Lambda_k \mid - ) = \mathcal{N}(\mu_{\Lambda_k}, \sigma^2_{\Lambda_k}) \cdot \one\{\Lambda_k \in \mathcal{A}_\Lambda \cap \mathcal{A}_k\}.
\end{equation*}

The full conditional distributions of the eigenvectors $\phi_k$ and $\psi_k$ are
\begin{align*}
  \log \pp(\phi_k \mid \{\psi_k, \Lambda_k\}_{k=1}^m, \phi_{-k}) &= -\frac{\nu}{2} \lVert{\bY_k - \Lambda_k \phi_k \psi_k^\T}\rVert_\mathrm{F}^2 - \frac{\norm{\psi_k}_2^2}{2\sigma_\psi^2} - \sum_{j=1}^m \frac{(\delta_{jk} - \psi_j^\T \phi_k)^2}{2\sigma_c^2}, \\
  \log \pp(\psi_k \mid \{\phi_k, \Lambda_k\}_{k=1}^m, \psi_{-k}) &= -\frac{\nu}{2} \lVert{\bY_k - \Lambda_k \phi_k \psi_k^\T}\rVert_\mathrm{F}^2 - \frac{\norm{\phi_k}_2^2}{2\sigma_\phi^2} - \sum_{j=1}^m \frac{(\delta_{jk} - \psi_k^\T \phi_j)^2}{2\sigma_c^2}.
\end{align*}
Define $\bphi = (\phi_1, \dots, \phi_m) \in \mathbb{R}^{m \times m}$ and $\bpsi = (\psi_1, \dots, \psi_m) \in \mathbb{R}^{m \times m}$ as the matrix of all $m$ of the eigenvectors, $\mathrm{I}_m$ as the $m \times m$ identity matrix, and $\bdelta_k = (\delta_{1k}, \dots, \delta_{mk})^\T$ as a vector with $k$th value $1$ and $0$ otherwise. The full conditional distribution of $\phi_k$ is thus given by a normal distribution $(\phi_k \mid - ) \sim \mathcal{N}(\mu_{\phi_k}, \Sigma_{\phi_k})$ with parameters
\begin{equation*}
  \mu_{\phi_k} = \Sigma_{\phi_k} \left( \nu \Lambda_k \bY_k \psi_k + \frac{\bpsi \bdelta_k}{\sigma_c^2} \right), \quad \text{and} \quad \Sigma_{\phi_k}^{-1} = \left( \frac{1}{\sigma_{\phi}^2} + \nu \norm{\Lambda_k \psi_k}_2^2\right) \mathrm{I}_m + \frac{\bpsi \bpsi^\T}{\sigma_c^2}.
\end{equation*}
Similarly, the full conditional distribution of $\psi_k$ is given by a normal distribution $(\psi_k \mid - ) \sim \mathcal{N}(\mu_{\psi_k}, \Sigma_{\psi_k})$ with parameters
\begin{equation*}
  \mu_{\psi_k} = \Sigma_{\psi_k} \left( \nu \Lambda_k \bY_k^\T \phi_k + \frac{\bphi \bdelta_k}{\sigma_c^2} \right), \quad \text{and} \quad \Sigma_{\psi_k}^{-1} = \left( \frac{1}{\sigma_{\psi}^2} + \nu \norm{\Lambda_k \phi_k}_2^2\right) \mathrm{I}_m + \frac{\bphi \bphi^\T}{\sigma_c^2}.
\end{equation*}
To maintain non-negativity of the off-diagonal entries of $\bL$, we impose truncation constraints on the components of $\phi_k$ and $\psi_k$. By analogous reasoning to the $\Lambda_k$, we define the element-wise truncation sets as
\begin{align*}
    \phi_k(q) &\in \left[ \underset{\lambda_k \psi_k(q) > 0}{\mathrm{max}} \frac{\bL_{-k}(k,q)}{\lambda_k \psi_k(q)}, \underset{\lambda_k \psi_k(q) < 0}{\mathrm{min}} \frac{\bL_{-k}(k,q)}{\lambda_k \psi_k(q)}\right], \\
    \psi_k(p) &\in \left[ \underset{\lambda_k \phi_k(p) > 0}{\mathrm{max}} \frac{\bL_{-k}(p,k)}{\lambda_k \phi_k(p)}, \underset{\lambda_k \phi_k(p) < 0}{\mathrm{min}} \frac{\bL_{-k}(p, k)}{\lambda_k \phi_k(p)}\right].
\end{align*}
The truncation sets depend on the values of all other eigenvectors and eigenvalues via $\bL_{-k}$ and are not fixed, but state dependent and must be recomputed at every Gibbs step. Sampling may proceed through a multivariate truncated normal or by sequentially sampling each component of $\phi_k$ and $\psi_k$ conditional on the others. For all truncations, the bounds can become narrow; in practice, a small relaxation term can be added to widen the sets and aid mixing. The choice of $\sigma_c^2$ determines the degree of shrinkage toward orthogonality. In practice, we choose $\sigma_c^2 \approx 0$ to obtain near-orthogonality. If exact orthogonality is required, the projection onto the Stiefel manifold is unique, so posterior samples may be projected as per \cite{astfalck2024posterior}.

\section{Theoretical justifications} \label{sec:justification}

In our proposed approach, inference proceeds in two steps: first, the transition probability matrix $\bP$ is estimated from the data (potentially ignoring its dependence on the spectral parameters); second, the spectral parameters are estimated from the resulting estimate of $\bP$. We now provide theoretical results to support this two-stage approach. We begin with a BvM theorem for discretely observed CTMCs, followed by a demonstration of posterior consistency for the spectral parameters. These results extend the existing theory on the consistency of the transition probability matrix in Markov chains. Proofs of all Theorems are included in the Appendix.

\subsection{Bernstein--von Mises Theorem for $\bP$}

Let $\bP \in [0, 1]^{m \times m}$ be the transition probability matrix of a regularly-observed CTMC with true value $\bP_0$, and define the local parameter $h = \sqrt{n}(\bP - \bP_0)$. We make the following two assumptions.
\begin{assumptionp}{(A1)} The true transition probability matrix $\bP_0$ satisfies $\bP_0(p,q) > c/m$ for all $p,q \in \{1, \dots, m\}$, for some constant $0 < c < 1$, and thus $\bP_0$ has full support.
\end{assumptionp}
\begin{assumptionp}{(A2)}
  The prior on $\bP$ is modified as
  \begin{equation*}
    \bP(p, \cdot) \sim \Dir(\balpha) \one\{\bP(p, \cdot) > c/m\}
  \end{equation*}
for some positive $\balpha \in \mathbb{R}_{+}^m$ and small constant $0 < c < 1$.
\end{assumptionp}

Assumption (A1) ensures that $\bP_0$ lies in the interior of $[0, 1]^{m^2}$. Assumptions (A1) and (A2) together guarantee recurrence and geometric ergodicity of the Markov chain, and ensure that $\bP$ is identifiable; further details are provided in the Appendix. We may now state the following result.

\begin{theorem} \label{the:bvm}
  Under Assumptions (A1) and (A2), the posterior distribution $\pp(\bP \mid \bx_n)$ satisfies the Bernstein--von Mises theorem,
  \begin{equation*}
    \mathbb{E}_{\pp(\bx_n \mid \bP_0)}\left\lVert \pp(\bP \mid \bx_n) - \mathcal{N}\left(\hat{\bP}_n, \frac{1}{n} \mathcal{I}^{-1}(\bP_0)\right)\right\rVert_{\mathrm{TV}} \rightarrow 0 \quad \text{as} \quad n \rightarrow \infty,
  \end{equation*}
  where $\hat{\bP}_n$ is the maximum likelihood estimator of $\bP$, $\mathcal{I}(\bP_0)$ is the Fisher information matrix at $\bP_0$, and $\norm{\cdot}_\mathrm{TV}$ denotes the total variation norm.
\end{theorem}

Theorem~\ref{the:bvm} builds on extensions of BvM theory to weakly dependent data, as developed in \cite{connault2014weakly}. In our setting, the dependence structure arises from the Markov property of the discretely observed CTMC. To establish this result, we verify that the assumptions required by \cite{connault2014weakly} hold under our setup. These include regularity of the likelihood, smoothness of the score function, and concentration inequalities for empirical averages over Markovian blocks. We now state the assumptions required to apply their result in our setting.

\begin{assumptionp}{(C1)}
  For each $n$, the likelihood function $\mathcal{L}(\bx_n \mid \bP_0)$ is dominated by a common $\sigma$-finite measure and $\mathcal{L}(\bx_n \mid \bP_0)$ denotes its density with respect to the common dominating measure.
\end{assumptionp}

\begin{assumptionp}{(C2)}
  The prior $\pp(\bP)$ is absolutely continuous with respect to the Lebesgue
  measure in a neighborhood of $\bP_0 \in \mathbb{R}^{m^2}$ and has a continuous, strictly positive density at $\bP_0$.
\end{assumptionp}

\begin{assumptionp}{(C3)}
  For every $\epsilon > 0$, there is a sequence of uniformly consistent tests $\varphi_n$ such that
  \begin{equation*}
    \mathbb{E}_{\bP_0}[\varphi_n] \rightarrow 0 \quad \text{and} \quad \sup_{\norm{\bP - \bP_0} \geq \epsilon} \mathbb{E}_\bP[1 - \varphi_n] \rightarrow 0
  \end{equation*}
  as $n \rightarrow \infty$.
\end{assumptionp}

\begin{assumptionp}{(C4)}
  There is a sequence of random variables $r_n = n^{-1/2} \nabla_{\bP_0} \log \mathcal{L}(\bx_n \mid \bP)$ and an invertible matrix $
  \mathcal{I}$, such that $r_n$ converges weakly to $\mathcal{N}(0, \mathcal{I})$ under $\bP_0$ such that for every sequence $h_n \rightarrow h$, we have
  \begin{equation*}
    \log \mathcal{L}(\bx_n \mid \bP_0 + n^{-1/2}h_n) - \log \mathcal{L}(\bx_n \mid \bP_0) = h_n^\T r_n - \frac{1}{2} h_n^\T \mathcal{I} h_n + o_{p}(1)
  \end{equation*}
  uniformly in $h_n$ in compact sets.
\end{assumptionp}

\begin{assumptionp}{(C5)}
  Let $s_n = n^{-1/2} r_n$, where $r_n$ is as defined in (C4). There exist $n_0 \in \mathbb{N}$, $\epsilon < 1$, and $c > 0$ such that for any $\norm{\bP - \bP_0} \leq \epsilon$, and $n > n_0$,
  \begin{equation*}
    \norm{\mathbb{E}_\bP(s_n) - \mathbb{E}_{\bP_0}(s_n)} \geq c \norm{\bP - \bP_0}.
  \end{equation*}
\end{assumptionp}

\begin{assumptionp}{(C6)}
  A two-sided McDiarmid-type concentration inequality holds for $s_n$, defined in (C5). That is, there exists $n_0 \in \mathbb{N}$ and $c < \infty$ such that for all $\bP$ and $n > n_0$
  \begin{equation*}
    \pp(\norm{s_n - \mathbb{E}_\bP[s_n]} > u)\leq 2 \exp\left\{\frac{-u^2}{2c/n}\right\}.
  \end{equation*}
\end{assumptionp}

\begin{assumptionp}{(C7)}
  A one-sided Hoeffding-type inequality holds for non-overlapping consecutive blocks of observations. Let $R \in \mathbb{N}$ and define blocks $X_j^R = \{X_i\}_{i = (j-1)R + 1}^{jR}$. For $g_s = g(X_s^R)$, there exists some $n_0 \in \mathbb{N}$ and $c < \infty$ such that for any measurable function $0 \leq g \leq 1$, any $\bP$, and all $n > n_0$,
  \begin{equation*}
    \pp\left(\frac{1}{S}\sum_{s_1}^S g_s < \mathbb{E}\left[\frac{1}{S}\sum_{s_1}^S g_s\right] - u\right) \leq \exp\left\{ \frac{-u^2}{2c/S}\right\}.
  \end{equation*}
\end{assumptionp}

Conditions (C1)--(C4) correspond to standard requirements in BvM theorems for independent and identically distributed data. Conditions (C5)--(C7) were introduced in \cite{connault2014weakly} to extend the classical BvM theorem to weakly dependent data. These assumptions provide control over the score function and block-wise averages through exponential inequalities, allowing the construction of uniformly consistent tests and establishing a local asymptotic normality expansion. In Appendix~\ref{sec:proofs} we provide proofs that Conditions (C1)--(C7) are satisfied by (A1) and (A2) and so Theorem~\ref{the:bvm} is proved as in \cite{connault2014weakly}.

\subsection{Posterior consistency of the spectral parameters}

The BvM theorem established in the preceding section ensures asymptotic normality and local concentration of the posterior distribution of $\bP$. We now focus on the asymptotic behavior of the spectral parameters.

\begin{theorem} \label{the:spectral_consistency}
  Define the spectral decomposition of the true generator $\bL_0$ as $\{\lambda_k^0, \phi_k^0, \psi_k^0\}_{k=1}^m$, and the spectral decomposition of the empirical generator $\hat{\bL} = \Delta^{-1} \log\{\hat{\bP}\}$ as $\{\hat{\lambda}_k, \hat{\phi}_k, \hat{\psi}_k\}_{k=1}^m$.
   For $\bx_n = \{x_0, x_1, \dots, x_n\}$ sampled at a rate $\Delta$, the spectral parameters of the generator satisfy
  \begin{align*}
    \sqrt{n} (\hat{\lambda}_k - \lambda_k^0) &\rightarrow \Delta^{-1} \exp\{- \lambda_k^0 \Delta\} (\psi_k^0)^\T \bQ \phi_k^0 \\
    \sqrt{n} (\hat{\psi}_k - \psi_k^0) &\rightarrow  \bR_k^\T \bQ^\T \psi_k^0   \\
    \sqrt{n} (\hat{\phi}_k - \phi_k^0) &\rightarrow  \bR_k \bQ \phi_k^0 \\
  \end{align*}
  in probability, where $\bQ \sim \mathcal{N}(0, 2 \mathcal{I}^{-1}(\bP_0))$ and $\bR_k = \sum_{j \ne k} \frac{\phi_j^0 (\psi_j^0)^\T}{\Lambda_k^0 - \Lambda_j^0},$ where $\Lambda_k^0 = \exp\{\lambda_k^0 \Delta\}$.
\end{theorem}

Theorem~2 characterizes how uncertainty in $\bP$ propagates to uncertainty in the spectral parameters $\{\lambda_k, \psi_k, \phi_k\}_{k=1}^m$. Notably, Theorem~2 establishes that the posterior distributions for the spectral parameters converge at the standard parametric rate of $\sqrt{n}$ and explicitly capture how factors such as sampling frequency and the eigenvalues affect convergence. The convergence of left and right eigenvectors is governed by $\bR_k$. When the eigenvalues are close, or concentrated together, $\bR_k$ increases, thus increasing the uncertainty in the eigenvector estimation. The intuitive rationale here is that when eigenvalues are clustered closely together, the corresponding eigenvectors are sensitive to nearby eigenmodes. From a theoretical perspective, Theorem~2 extends the asymptotic normality results on $\bP$ in Theorem~1 to the derived spectral parameters.

\section{Applications} \label{sec:applications}

We now demonstrate our approach through two studies. First, we compare to existing methodologies across repeated simulation; namely, \cite{bladt2005statistical} and \cite{riva2023estimation}. Second, we scale our method to an application of a metastable diffusion process.

\subsection{Simulation of Random Generators}

We begin with a simulation study, considering CTMCs with observation lengths $n \in \{10^2, 10^3, 10^4, 10^5\}$ and state space dimension $m \in \{2, 4, 8\}$. We refer to our method as \textsc{bigmac} (Biorthogonal Inference for the Generator Matrix of A CTMC), the method of \cite{bladt2005statistical} as \textsc{blandt}, and the method of \cite{riva2023estimation} as \textsc{mh-riva}. For each combination of $n$ and $m$, we perform 100 simulations using a randomly sampled generator matrix $\bL$, whose off-diagonal elements are drawn as
\begin{equation*}
  \bL(p, q) = \bL(q, p) \sim \frac{2}{|p - q|^2} \mathcal{U}(0, 1) \quad \text{for all} \quad p < q,
\end{equation*}
where $\mathcal{U}(0, 1)$ denotes the uniform distribution on $[0,1]$. The diagonal entries are set as $\bL(p,p) = - \sum_{q \neq p} \bL(p,q)$, ensuring each row sums to zero. This construction yields symmetric generator matrices with a loosely banded structure.

For \textsc{bigmac}, we set the prior hyperparameters to $\nu = 10^4$, $\sigma_\phi^2 = 10^{-1}$, $\sigma_\psi^2 = 10^{-1}$, $\sigma_c^2 = 10^{-5}$, and $\balpha = (1, \dots, 1)$. 
The parameters $\sigma_\phi^2$ and $\sigma_\psi^2$ aid identifiability of the spectral components, and the product $\nu \sigma_c^2$ governs the strength of shrinkage toward the biorthogonal structure. When $\nu \sigma_c^2 = 1$, the model gives equal influence to representing the transition matrix $\bP$ (regardless of biorthogonality) and enforcement of biorthogonality; values less than one, as used here, favor biorthogonal solutions. We implement \textsc{blandt} using the \texttt{ctmcd} package \citep{pfeuffer2017ctmcd}, which augments \citet{bladt2005statistical} with the latent-path sampling speed-up of \citet{fearnhead2006exact}. We adopt the default uninformative priors provided by the package. Finally, we implement \textsc{mh-riva}, the Metropolis--Hastings sampler of \citet{riva2023estimation}. Here, $\bL$ is parameterized as $\bL = -\bA + \bA \bQ$, where $\bA$ is a positive diagonal matrix with entries $a_{ii} > 0$, and $\bQ$ is a probability matrix with a zero-valued diagonal (note, \emph{a} probability matrix, not \emph{the} transition matrix $\bP$). Proposals for $\bA$ are made element-wise as $\log a^{\mathrm{proposed}}_{ii} \sim \mathcal{N}(\log a^{\mathrm{current}}_{ii}, \sigma_a^2)$. For $\bQ$, we define $\bQ_{j, -j}$ as the $j$th row excluding the diagonal entry, and generate proposals as $\bQ_{j, -j}^{\mathrm{proposed}} \sim \mathrm{Dir}(c \bQ_{j, -j}^{\mathrm{current}} + 1)$. We set $\sigma_a^2 = 3n^{-1/2}$ and $c = n^{1/2}$, which we find provides a good exploration--exploitation balance that scales with $n$ as the posterior contracts.

\begin{figure}[b!]
  \centering
  \includegraphics[width=0.6\textwidth]{"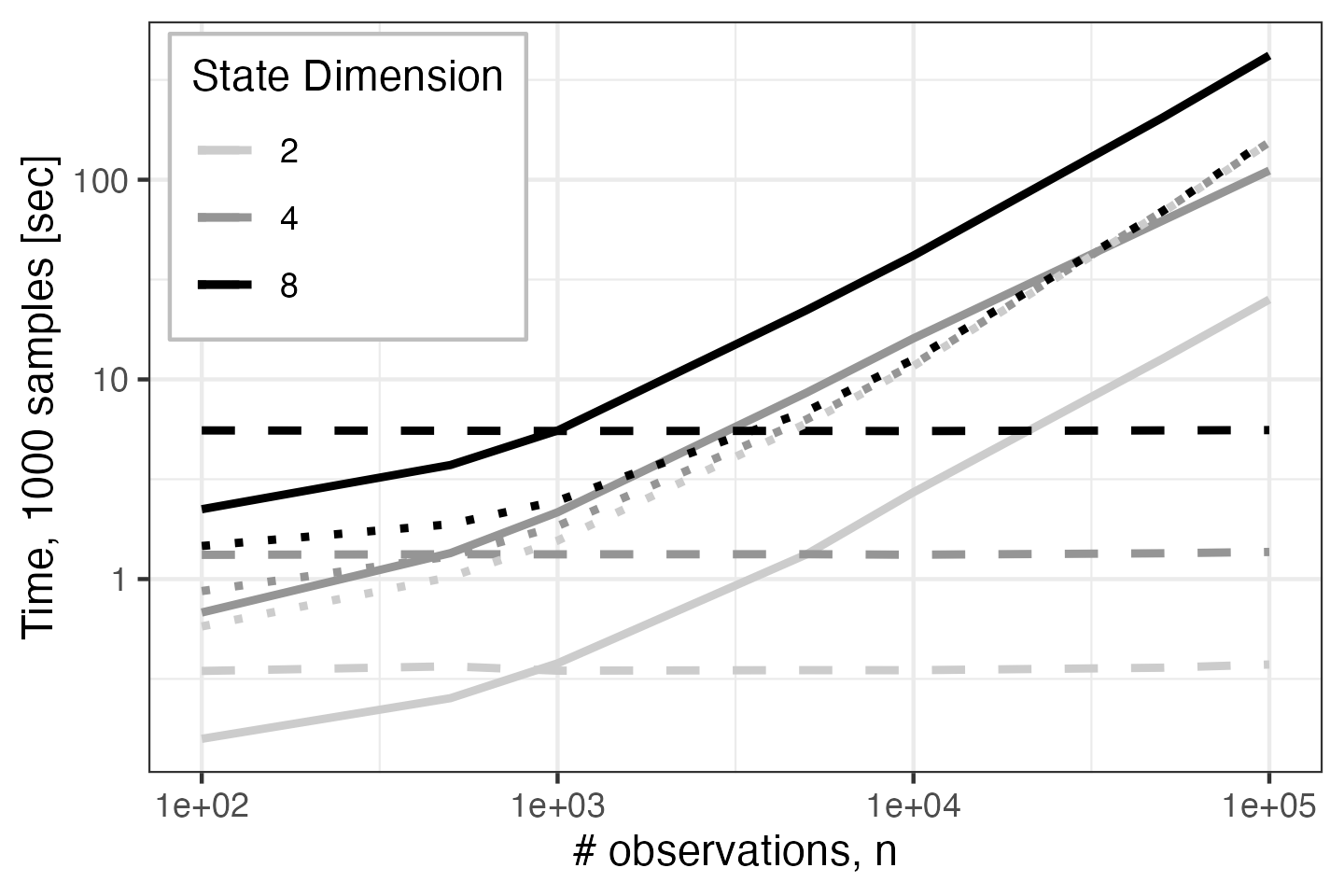"}
  \caption{Run-times per 1000 MCMC samples. The model \textsc{blandt} is shown with solid lines (\textbf{---}), \textsc{bigmac} with dashed lines (\textbf{- - -}), and \textsc{mh-riva} with dotted lines ($\bm{\cdot \cdot \cdot}$).}
  \label{fig:run_times}
\end{figure}

First, we plot in Figure~\ref{fig:run_times} the run-times for each of the three models over $n$; line shading denotes the state space dimension $m$. The models are distinguished by line type: solid lines (\textbf{---}) for \textsc{blandt}, dashed lines (\textbf{- - -}) for \textsc{bigmac}, and dotted lines ($\bm{\cdot \cdot \cdot}$) for \textsc{mh-riva}. All models are run on an Apple M2 Pro chip with 32GB RAM; timings should be interpreted as relative, since absolute run-times may vary across systems. Both \textsc{blandt} and \textsc{mh-riva} scale computationally with increasing $n$. For \textsc{blandt}, this is due to latent path sampling, while for \textsc{mh-riva}, the increase stems from repeated likelihood evaluations. In contrast, \textsc{bigmac} only requires a single computation of a transition count matrix. With increasing state space dimension $m$, both \textsc{bigmac} and \textsc{blandt} incur higher computational cost. Interestingly, \textsc{mh-riva} remains computationally consistent in $m$, as proposals are fast and the cost of likelihood evaluation scales with $n$ rather than $m$. Despite this, we observe severe performance issues for effective sample size (ESS). Table~\ref{tab:ess} reports the average ESS per 100 samples, where an aggregated ESS is computed as the entrywise mean ESS across $\bL$, for $n = 10^3$ and $n = 10^5$, and for each $m$. Standard deviations across simulations are shown in parentheses. In general, ESS decreases with increasing $m$, as expected due to the growing parameter space. The models \textsc{bigmac} and \textsc{blandt} exhibit similar mixing properties, making their run-times in Figure~\ref{fig:run_times} directly comparable. By contrast, \textsc{mh-riva} suffers from markedly lower ESS and deteriorates as both $n$ and $m$ increase. While improved mixing may be achievable through more refined proposal tuning, this would require considerable manual effort on a case-by-case basis.

\begin{table}[ht]
  \centering
  \caption{Average effective sample size per 100 samples, computed as the entrywise mean effective sample size across the elements of $\bL$.}
  \begin{tabular}{ccccccc}
  \toprule
  & \multicolumn{3}{c}{$n = 10^3$} & \multicolumn{3}{c}{$n = 10^5$} \\
   & $m = 2$ & $m = 4$ & $m = 8$ & $m = 2$ & $m = 4$ & $m = 8$ \\
  \midrule
  \textsc{bigmac} & 72 (30) & 22 (11) & 8 (2) & 26 (23) & 8.7 (4.8) & 6.5 (1.9) \\
  \textsc{blandt} & 56 (47) & 12 (8) & 8 (2) & 52 (36) & 16 (8.8) & 3.7 (1.3) \\
  \textsc{mh-riva} & 7.4 (2.6) & 1.5 (0.37) & 0.54 (1.4) & 8.0 (2.9) & 0.17 (0.42) & 0 (0) \\
  \bottomrule
  \end{tabular}
  \label{tab:ess}
\end{table}

Finally, we examine the goodness-of-fit of the models. As noted above, \textsc{mh-riva} does not yield a sufficient number of effectively independent samples, and is therefore excluded from this comparison. Table~\ref{tab:rmse_fits} reports the average Frobenius norm between the posterior means and the true $\bL$, averaged across all simulations. As in Table~\ref{tab:ess}, bracketed values denote empirical standard deviations. Since \textsc{blandt} targets the exact posterior, we treat its results as a proxy for ground truth. Our model \textsc{bigmac} exhibits comparable performance across all values of $m$ and $n$, empirically suggesting that the pseudo-likelihood underpinning our approach is adequate. Combined with the computational results above, this supports our claim that we achieve comparable accuracy with significantly improved efficiency.

\begin{table}[ht]
  \centering
  \caption{Average Frobenius norms between the posterior means and the true $\bL$, computed across all simulations. Bracketed values denote empirical standard deviations.}
  \begin{tabular}{cccccc}
  \toprule
  & & $n = 10^2$ & $n = 10^3$ & $n = 10^4$ & $n = 10^5$ \\
  \midrule
  \multirow{2}{*}{$m=2$} & \textsc{bigmac} & 0.54 (0.82) & 0.080 (0.073) & 0.026 (0.019) & 0.012 (0.0085) \\
   & \textsc{blandt} & 0.27 (0.21) & 0.085 (0.077) & 0.027 (0.021) & 0.0094 (0.0083) \\
   \midrule
   \multirow{2}{*}{$m=4$} & \textsc{bigmac} & 0.99 (0.47) & 0.38 (0.22) & 0.14 (0.064) & 0.080 (0.032) \\
   & \textsc{blandt} & 1.0 (0.39) & 0.47 (0.21) & 0.16 (0.79) & 0.050 (0.029) \\
   \midrule
   \multirow{2}{*}{$m=8$} & \textsc{bigmac} & 2.0 (0.92) & 0.71 (0.18) & 0.43 (0.11) & 0.35 (0.10) \\
   & \textsc{blandt} & 2.2 (0.49) & 1.4 (0.33) & 0.62 (0.19) & 0.25 (0.086) \\
  \bottomrule
  \end{tabular}
  \label{tab:rmse_fits}
\end{table}

\subsection{Metastable Diffusion Process}

We next consider a higher-dimensional example exhibiting metastable behavior, based on a continuous-time diffusion model. The Smoluchowski equation defines a class of irreducible Markov processes governed by the stochastic differential equation
\begin{equation} \label{eqn:smoluchowski}
\rd x_t = -\frac{1}{\alpha} \nabla \mathrm{U}(x_t) \rd t + \frac{\beta}{\alpha} \rd w_t,
\end{equation}
where $\alpha$ and $\beta$ are known parameters, and $\{w_t: t \ge 0\}$ denotes standard Brownian motion \citep[see][]{huisinga2004phase}. Following the parameterization in \cite{deng2012model}, we fix $\alpha = 1.5$ and $\beta = 1.65$, and define the potential $\mathrm{U} : \mathbb{R} \to \mathbb{R}^+$ as the sixth-degree polynomial
\begin{equation*}
\mathrm{U}(x) = \frac{1}{200} \left( 0.5x^6 - 15x^4 + 119x^2 + 28x + 50 \right),
\end{equation*}
which forms a three-well potential, shown in the left panel of Figure~\ref{fig:diffusion}. This choice of $\mathrm{U}(x)$ yields a metastable stochastic process, in which $x_t$ remains trapped for extended periods within one of three dominant basins of attraction, punctuated by rare and rapid transitions between them \citep[see][for a rigorous formulation]{huisinga2004phase}. We simulate a single trajectory from this process using an Euler--Maruyama scheme with step size $0.02$ over $5 \times 10^5$ time steps. Observations are obtained by subsampling every 50th point, yielding $10{,}000$ discrete-time observations of the continuous process. To perform inference, we discretize the continuous state space $[-5, 5]$ into $m = 30$ equal-width intervals and project the simulated trajectory onto this discrete partition. The resulting sequence of interval indices forms a discrete-state approximation to the original diffusion. Under standard regularity conditions, this projection induces a CTMC over the discretized state space, with a generator matrix that approximates the coarse-grained dynamics of the latent diffusion. Coarse-graining procedures of this type are widely used in statistical physics and molecular dynamics to approximate high-dimensional diffusions by finite-state Markov models \citep[see][]{saunders2013coarse,castiglione2008chaos}. The right panel of Figure~\ref{fig:diffusion} shows the empirical transition matrix $\hat{\bP}$ obtained from the subsampled trajectory. For reference, the center panel displays a higher-fidelity approximation of the transition matrix, computed analogously using a long-run simulation of $5 \times 10^7$ time steps.

\begin{figure}[h!]
  \centering
  \includegraphics[width=\textwidth]{"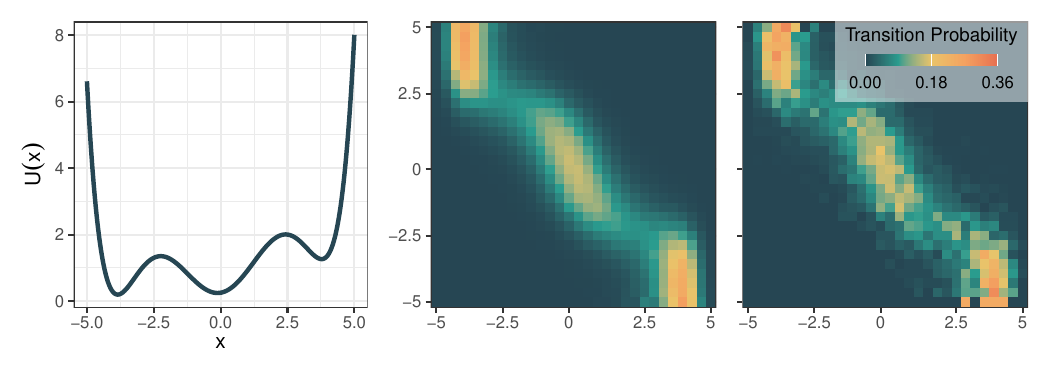"}
  \caption{Left, the three-well potential $\mathrm{U}(x)$ governing the dynamics in the Smoluchowski equation. Center, reference transition matrix estimated from a long simulation of $5 \times 10^7$ time steps. Right, empirical transition matrix $\hat{\bP}$ estimated from a subsampled trajectory of $10{,}000$ observations.}
  \label{fig:diffusion}
\end{figure}

\begin{figure}[t!]
  \centering
  \includegraphics[width=\textwidth]{"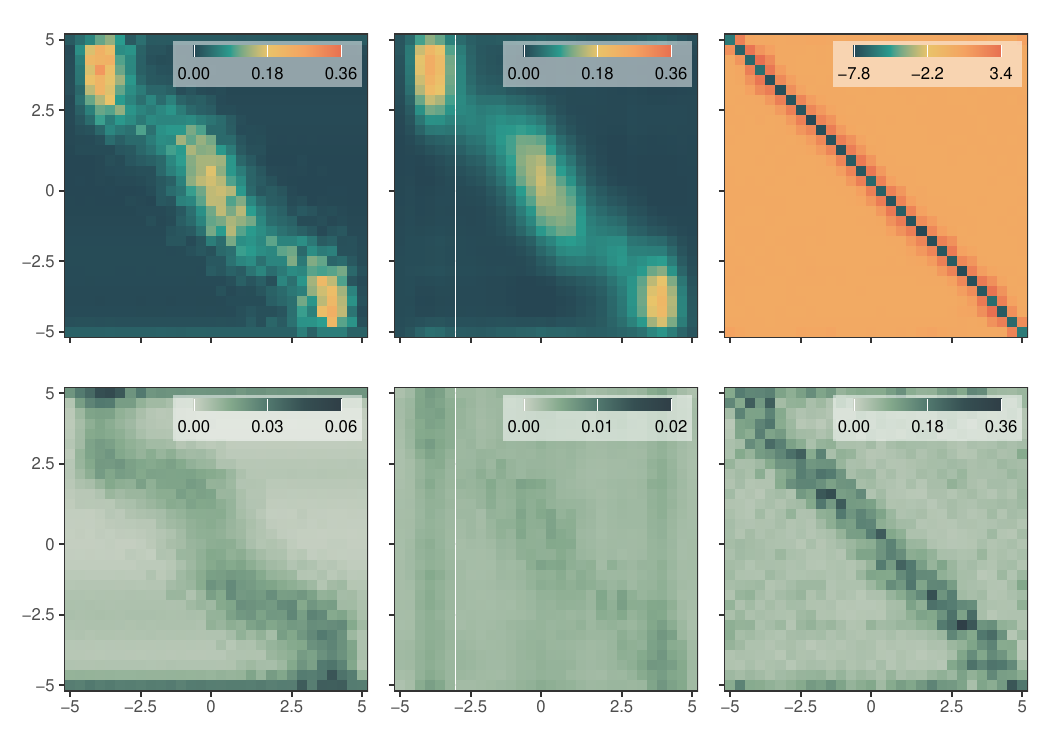"}
  \caption{
  Top row, posterior means of the transition matrix $\bP$ (left), its spectral reconstruction $\tilde{\bP}$ (center), and the generator $\bL$ (right). Bottom row, corresponding posterior standard deviations.
  }
  \label{fig:diffusion_fits}
\end{figure}

\begin{figure}[h!]
  \centering
  \includegraphics[width=\textwidth]{"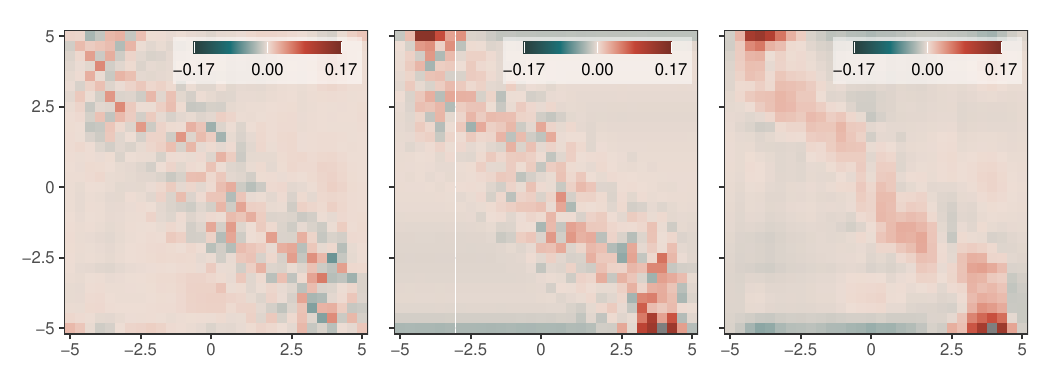"}
  \caption{Left, expected difference $\E[\bP - \tilde{\bP} \mid \bx_n]$. Center, difference between $\E[\bP \mid \bx_n]$ and the reference matrix. Right, difference between $\E[\tilde{\bP} \mid \bx_n]$ and the reference matrix.}
  \label{fig:difference}
\end{figure}

This combination of dimensionality and observation length is challenging for \textsc{blandt} and \textsc{mh-riva} to model. In contrast, we demonstrate that our model scales effectively and accurately recovers the dynamics of the diffusion process. As in the previous case study, we set the prior hyperparameters to $\nu = 10^4$, $\sigma_\phi^2 = 10^{-1}$, $\sigma_\psi^2 = 10^{-1}$, $\sigma_c^2 = 10^{-5}$, and $\balpha = (1, \dots, 1)$. Summaries of the model fit are shown in Figure~\ref{fig:diffusion_fits}. The top row displays the posterior means, and the bottom row the posterior standard deviations, for the transition probability matrix $\bP$ (left), the reconstructed transition probability matrix $\tilde{\bP} = \sum_k \exp\{\lambda_k\} \phi_k \psi_k^\T$ (center), and the generator matrix $\bL$ (right). The posterior mean of $\bP$ closely matches the reference matrix from the long-run simulation, with most mass concentrated along near-diagonal transitions between neighboring bins. As we place a weak prior on each row $\bP(i, \cdot)$, the posterior mean is nearly identical to the empirical plug-in estimator $\hat{\bP}$ shown earlier in Figure~\ref{fig:diffusion}. The reconstructed matrix $\tilde{\bP}$ captures similar structure, but exhibits smoother variation between adjacent entries, reflecting the regularizing influence of the biorthogonality constraint. To assess structural differences, Figure~\ref{fig:difference} shows the expected difference $\E[\bP - \tilde{\bP} \mid \bx_n]$ (left), the deviation between the reference matrix and $\E[\bP \mid \bx_n]$ (center), and the deviation between the reference matrix and $\E[\tilde{\bP} \mid \bx_n]$ (right). At first glance, the right-hand panel suggests structural bias in $\tilde{\bP}$. However, as seen in the left panel, this difference aligns closely with variations already present in $\bP$, indicating that $\tilde{\bP}$ provides a smoothed and plausible reconstruction rather than introducing systematic error. Posterior standard deviations (bottom row of Figure~\ref{fig:diffusion_fits}) are largest in the transition regions at basin boundaries, where transitions are rare but not negligible. Notably, uncertainty in $\tilde{\bP}$ is smaller overall (note the different legend scale), again reflecting regularization through the spectral approximation. The final column in Figure~\ref{fig:diffusion_fits} shows the posterior summaries for the generator matrix $\bL$. The posterior mean exhibits dominant diagonal structure and sparse off-diagonal mass. Posterior uncertainty is concentrated along the diagonal, while off-diagonal transition regions, where rates are near zero, exhibit greater certainty.

\section{Conclusions} \label{sec:conclusion}

We present a pseudo-Bayesian framework for inference on CTMCs observed at discrete time points, addressing the challenge posed by the intractability of the true likelihood. By constructing a pseudo-likelihood over both the transition matrix $\bP$ and its spectral generator $\bL$, we offer a principled solution to the embedding problem that preserves probabilistic structure while enabling scalable inference. The resulting Gibbs sampler respects the constraints required for embeddability and yields coherent posterior uncertainty. Our method is straightforward to implement, computationally efficient, and theoretically grounded. We prove a BvM theorem for the transition matrix and establish posterior consistency for the spectral decomposition of $\bL$. Through empirical studies on synthetic and diffusion-based systems, we demonstrate the method's capacity to recover and quantify uncertainty in the dynamics. Our approach relies on two main assumptions: regular observation times and a biorthogonal generator $\bL$. Future work could explore relaxing one or both of these assumptions.

\section*{Acknowledgements}

Lachlan Astfalck was supported by the ARC TIDE Industrial Transformation Research Hub (Grant No. IH200100009). David Dunson was partially supported by the Office of Naval Research (N00014-24-1-2626) and the National Science Foundation (NSF IIS-2426762).

\section*{Code and data availability}

Code is available at \texttt{astfalckl.github.io/bigmac} as an R package. Tutorials are provided to recreate the simulation studies and code is given to generate the data used herein.

\bibliographystyle{biometrika}
\bibliography{references}

\appendix

\newpage

\spacingset{1.1}

\section{Gibbs Sampler Pseudo-code} \label{sec:gibbs}

\begin{algorithm}[ht]
  \caption{Gibbs sampler for pseudo-posterior inference of CTMCs}
  \begin{algorithmic}[1]
  \State \textbf{Input:} Data $\bx_n$, hyperparameters $\balpha$, $\nu$, $\Delta$, $\sigma_\phi^2$, $\sigma_\psi^2$, $\sigma_c^2$
  \State \textbf{Initialize:} $\bP^{(0)}$, $\{\Lambda_k^{(0)}, \phi_k^{(0)}, \psi_k^{(0)}\}_{k=1}^m$; fix $\phi_1 = \mathbf{1}$ and $\Lambda_1 = 1$. Compute counts $\mathbf{c}_p = (c_p^1, \dots, c_p^m)$ from transitions in $\bx_n$
  
  \For{iteration $t = 1$ to $N$}
    \For{each row $p = 1, \dots, m$}
      \State Sample $\bP^{(t)}(p, \cdot) \sim \Dir(\balpha + \mathbf{c}_p)$
      \Comment{Or use as proposal in MH if desired}
    \EndFor
  
    \For{each $k = 2, \dots, m$}
      \State \textbf{Update eigenvalue} $\Lambda_k$:
      \State $\bY_k \gets \bP^{(t)} - \sum_{j \ne k} \Lambda_j^{(t-1)} \phi_j^{(t-1)} (\psi_j^{(t-1)})^{\T}$
      \State Compute $\mu_{\Lambda_k}$ and $\sigma_{\Lambda_k}^2$
      \State Compute truncation intervals $\mathcal{A}_\Lambda$ (ordering) and $\mathcal{A}_k$ (positivity)
      \State Sample $\Lambda_k^{(t)} \sim \mathcal{N}(\mu_{\Lambda_k}, \sigma^2_{\Lambda_k}) \cdot \mathbbm{1}\{\Lambda_k \in \mathcal{A}_\Lambda \cap \mathcal{A}_k\}$
  
      \State \textbf{Update right eigenvector} $\phi_k$:
      \State Recompute $\bY_k$ using $\Lambda_k^{(t)}$
      \State Compute $\Sigma_{\phi_k}^{-1} = \left(\frac{1}{\sigma_\phi^2} + \nu \|\Lambda_k \psi_k\|^2 \right) \bI + \frac{\bpsi \bpsi^\T}{\sigma_c^2}$
      \State Compute $\mu_{\phi_k} = \Sigma_{\phi_k} \left(\nu \Lambda_k \bY_k \psi_k + \frac{\bpsi \bdelta_k}{\sigma_c^2} \right)$
      \State Compute elementwise truncation bounds $\mathcal{B}_k(j)$
      \State Sample $\phi_k^{(t)}(j) \sim \mathcal{N}(\mu_{\phi_k}, \Sigma_{\phi_k}) \cdot \mathds{1}\{\phi_k(j) \in \mathcal{B}_k(j)\}$
  \EndFor
  \For{each $k = 1, \dots, m$}
      \State \textbf{Update left eigenvector} $\psi_k$:
      \State Compute $\Sigma_{\psi_k}^{-1} = \left(\frac{1}{\sigma_\psi^2} + \nu \|\Lambda_k \phi_k\|^2 \right) \bI + \frac{\bphi \bphi^\T}{\sigma_c^2}$
      \State Compute $\mu_{\psi_k} = \Sigma_{\psi_k} \left(\nu \Lambda_k \bY_k^\T \phi_k + \frac{\bphi \bdelta_k}{\sigma_c^2} \right)$
      \State Compute elementwise truncation bounds $\mathcal{C}_k(j)$
      \State Sample $\psi_k^{(t)}(j) \sim \mathcal{N}(\mu_{\psi_k}, \Sigma_{\psi_k}) \cdot \mathds{1}\{\psi_k(j) \in \mathcal{C}_k(j)\}$
    \EndFor
  \EndFor
  
  \State \textbf{Output:} Posterior samples $\{\bP^{(t)}, \Lambda_k^{(t)}, \phi_k^{(t)}, \psi_k^{(t)}\}_{t=1}^N$
  \end{algorithmic}
  \end{algorithm}

  \newpage

\section{Technical Proofs} \label{sec:proofs}

\subsection{Supporting Lemmas}

To establish conditions (C1)--(C7) we first establish a series of lemmas that will assist.

\begin{lemma} \label{lem:concentration_inequalities}
  Let $\{X_i\}$ be a discrete-time Markov chain, and let $\tau(\epsilon)$ be the $\epsilon$-mixing time. Let $\tau_\mathrm{min} = \inf_{0 \leq \epsilon < 1} \tau(\epsilon)(\frac{2-\epsilon}{1-\epsilon})^2$. Let $f$ be a function of $n$ arguments with bounded differences
  \begin{equation*}
    f(X_1, \dots, X_n) - f(X_1', \dots, X_n') \leq \sum_{i=1}^n c_i \one \{X_i \ne X_i'\}.
  \end{equation*}
  Then, the following one-sided and two-sided concentration inequalities hold,
  \begin{align*}
    P(f \leq \E[f] - u) &\leq \exp\left\{\frac{-u^2}{2 \tau_\mathrm{min} \sum_{i=1}^n c_i^2}\right\} \\
    P(f \geq \E[f] + u) &\leq \exp\left\{\frac{-u^2}{2 \tau_\mathrm{min} \sum_{i=1}^n c_i^2}\right\}  \\
    P(\norm{f - \E[f]} \geq u) &\leq 2\exp\left\{\frac{-u^2}{2 \tau_\mathrm{min} \sum_{i=1}^n c_i^2}.\right\} 
  \end{align*}
\end{lemma}
\begin{proof}
  See Corollary 2.11 of \cite{paulin2015concentration}
\end{proof}

\begin{lemma} \label{lem:mixing}
  Define $\mathcal{P}$ as the compact set of irreducible, aperiodic transition matrices. Then there exists a constant $\bar{\tau}_{\mathrm{min}} < \infty$ such that
  \begin{equation*}
  \tau_{\mathrm{min}}(\bP) \le \bar{\tau}_{\mathrm{min}} \quad \text{for all } \bP \in \mathcal{P}.
  \end{equation*}
\end{lemma}

\begin{proof}
By Equation 3.12 in \cite{paulin2015concentration}, the mixing time satisfies,
\begin{equation*}
\tau_{\mathrm{min}}(\bP) \le 4 \cdot \frac{1 + 2\log(2) + \log(1/\mu_{\min}(\bP))}{1 - \lambda(\bP)},
\end{equation*}
where $\mu_{\min}(\bP) := \min_{x} \mu_x(\bP)$ is the smallest stationary probability, $\lambda(\bP)$ is the second-largest eigenvalue of the multiplicative reversibilization of $\bP$, as defined in \cite{fill1991eigenvalue}. Since $\mathcal{P}$ is compact, and both $\mu(\bP)$ and $\lambda(\bP)$ depend continuously on $\bP$, it follows that $\mu_{\min}(\bP)$ is bounded away from zero, and $\lambda(\bP)$ is bounded away from one. Thus, the right hand side of the inequality is uniformly bounded above over $\bP \in \mathcal{P}$, yielding a finite constant $\bar{\tau}_{\mathrm{min}}$ such that
\begin{equation*}
\tau_{\mathrm{min}}(\bP) \le \bar{\tau}_{\mathrm{min}} \quad \text{for all } \bP.
\end{equation*}
\end{proof}

\begin{corollary}
  We can extend Lemma~\ref{lem:mixing} to a block process. Define the block sequence $\hat{X}_j^R = \{X_i\}_{i =(j-1)R+1}^{jR}$. Since the original chain is uniformly mixing, the blocked chain $\{\hat{X}_j^R\}$ is also geometrically ergodic with the same rate, and thus has a uniform mixing time upper bound $\bar{\tau}_{\mathrm{min}, \hat{X}}$.  
\end{corollary}

\begin{lemma} \label{lem:uniform_ll}
  Let $\{X_i\}$ be a strictly stationary and ergodic Markov chain with transition probability matrix $\bP \in \mathcal{P}$ and initial distribution $X_0 \sim \mu$, where $\mu$ is the stationary distribution. Suppose $\mathcal{P}$ is compact and Assumption (A1) holds. Then,
  \begin{equation*}
  \sup_{\bP \in \mathcal{P}} \left| \frac{1}{n} \sum_{i=1}^n \log \bP(X_i \mid X_{i-1}) - \mathbb{E}_{\bP_0}\left[\log \bP(X_1 \mid X_0)\right] \right| \xrightarrow{\pp_{\bP_0}} 0,
  \end{equation*}
  \begin{equation*}
  \sup_{\bP \in \mathcal{P}} \left\| \frac{1}{n} \sum_{i=1}^n \nabla^2_\bP \log \bP(X_i \mid X_{i-1}) - \mathbb{E}_{\bP_0}\left[\nabla^2_\bP \log \bP(X_1 \mid X_0)\right] \right\| \xrightarrow{\pp_{\bP_0}} 0.
  \end{equation*}
\end{lemma}
  
\begin{proof}
  By the Markov property, the log-likelihood over $n$ observations is:
  \begin{equation*}
  \log \mathbb{L}(\{X_i\} \mid \bP) = \sum_{i=1}^n \log \bP(X_i \mid X_{i-1}),
  \end{equation*}
  and so the average log-likelihood is an empirical average of stationary, ergodic observations. From Assumption (A1), we have $\bP(p, q) > c/m$ for all $(p, q)$, implying that the function $\bP \mapsto \log \bP(X_i \mid X_{i-1})$ is continuous and uniformly bounded over $\mathcal{P}$. Since the state space is finite and $\mathcal{P}$ is compact, the family $\{\log \bP(x \mid x') : \bP \in \mathcal{P}\}$ is equicontinuous. Then, by the Uniform Ergodic Theorem \cite[e.g. Theorem 2.1 of][]{krengel1985ergodic}, we obtain,
  \begin{equation*}
  \sup_{\bP \in \mathcal{P}} \left| \frac{1}{n} \sum_{i=1}^n \log \bP(X_i \mid X_{i-1}) - \mathbb{E}_{\bP_0}\left[\log \bP(X_1 \mid X_0)\right] \right| \xrightarrow{\pp_{\bP_0}} 0.
  \end{equation*}

  The same argument applies to the observed information term, as $\nabla^2_\bP \log \bP(X_i \mid X_{i-1})$ is also continuous and uniformly bounded over $\mathcal{P}$ under (A1), completing the proof.
\end{proof}
  
\begin{lemma} \label{lem:lan_clt}
  Let $\{X_i\}$ be a strictly stationary and ergodic Markov chain observed at discrete times. Under assumption (C1), the score function satisfies
  \begin{equation*}
    \Delta_n = \frac{1}{\sqrt{n}} \nabla_{\bP_0} \log \mathbb{L}(\{X_i\}_{i=1}^{n} \mid \bP) \xrightarrow{\mathcal{D}} \mathcal{N}(0, \mathcal{I}^{-1}(\bP_0)),
  \end{equation*}
  where $\mathcal{I}(\bP_0)$ is the Fisher information matrix.
\end{lemma}

\begin{proof}
  Define the individual score contributions as $ V_i := \nabla_{\bP_0} \log \bP(X_i \mid X_{i-1}) $. Under (A1) and (C1), the chain is strictly stationary and ergodic, so $\{V_i\}$ is a stationary and ergodic sequence. Let $\mathcal{F}_{i} = \sigma(X_0, \dots, X_i)$ be the natural filtration. We now show that $\{V_i\}$ is a martingale difference sequence with respect to $\{\mathcal{F}_i\}$. Since $X_i \mid X_{i-1}$ is conditionally independent of the past given $X_{i-1}$, we compute
  \begin{align*}
    \mathbb{E}_{\bP_0}[V_i \mid \mathcal{F}_{i-1}]
    &= \mathbb{E}_{\bP_0}\left[\nabla_{\bP_0} \log \bP(X_i \mid X_{i-1}) \mid X_{i-1}\right] \\
    &= \nabla_{\bP_0} \sum_{x_i} \bP(x_i \mid X_{i-1}) \log \bP(x_i \mid X_{i-1}) = 0,  
  \end{align*}
  where the final equality follows from the fact that the derivative of the entropy is zero at the true parameter. Thus, $\{V_i\}$ is a martingale difference sequence. Therefore, the central limit theorem for stationary ergodic martingale differences \citep[see Theorem 27.4][]{billingsley1995probability} gives 
  \begin{equation*}
    \frac{1}{\sqrt{n}} \sum_{i=1}^{n} V_i \xrightarrow{\mathcal{D}} \mathcal{N}(0, \mathcal{I}^{-1}(\bP_0)).
  \end{equation*}
\end{proof}

\subsection{Proofs of Conditions (C1)--(C7)} \label{sec:conditions}

\begin{proof}[Proof of Condition (C1)]
Assumption (A1) ensures that $\bP_0$ satisfies $\bP_0(p, q) > c/m$ for some constant $0 < c < 1$ with states $p, q \in \{1, \dots, m\}$. For fixed $n$, the sample space is finite, so the likelihood $\mathcal{L}(\bx_n \mid \bP_0)$ is a density with respect to the product counting measure. Thus, $\mathcal{L}(\bx_n \mid \bP_0)$ is dominated by a common $\sigma$-finite measure, and (C1) is satisfied.
\end{proof}

\begin{proof}[Proof of Condition (C2)]
Assumption (A2) specifies that each row of $\bP$ is drawn from a truncated Dirichlet distribution,
\[
\bP(p, \cdot) \sim \mathrm{Dir}(\balpha) \one\{ \bP(p, \cdot) > c/m \},
\]
for some small $0 < c < 1$, truncating the Dirichlet distribution to the interior of the probability simplex. Thus, the prior is absolutely continuous with respect to the Lebesgue measure, and has a continuous and strictly positive density. From (A1), $\bP_0(p,q) > c/m$, so $\bP_0$ lies in the interior of the support of the prior. Therefore, the prior is absolutely continuous with respect to the Lebesgue measure in a neighborhood of $\bP_0$, and the density is continuous and strictly positive at $\bP_0$.
\end{proof}

\begin{proof}[Proof of Condition (C3)]
We verify the existence of uniformly consistent tests by constructing a uniformly consistent estimator of the transition matrix $\bP$. Define the sequence of consecutive, non-overlapping blocks $\hat{X}_i^2 = (X_{2i}, X_{2i+1})$, and let $\pi$ denote the stationary distribution of the block process $\{\hat{X}_i^2\}$. Then, each transition probability satisfies
\[
\bP(p,q) = \frac{\pi(p,q)}{\sum_{q'} \pi(p,q')}.
\]
We define the empirical estimator of $\pi$ via the empirical frequency of block occurrences,
\[
\hat{\pi}_n(\hat{x}) = \frac{1}{n} \sum_{i=1}^{n} \one\{ \hat{X}_i^2 = \hat{x} \}.
\]
By the ergodic theorem, $\hat{\pi}_n \to \pi$ almost surely as $n \to \infty$. We now show that this convergence is uniform in total variation over $\bP \in \mathcal{P}$.

Let $\mathcal{P}$ denote the compact set of irreducible, aperiodic transition matrices with full support, as guaranteed by (A1). By Lemma~\ref{lem:mixing} and its corollary, the block process $\{\hat{X}_i^2\}$ is uniformly geometrically ergodic with a uniform bound on its mixing time. We note that $\norm{ \hat{\pi}_n - \pi }_\mathrm{TV}$ satisfies the bounded differences condition in Lemma~\ref{lem:concentration_inequalities}, such that $\sum_{i=1}^n c_i^2 \leq \frac{4}{n}$. Thus, we obtain the concentration inequality
\[
\pp\left( \norm{ \hat{\pi}_n - \pi }_\mathrm{TV} \geq \E[ \norm{ \hat{\pi}_n - \pi }_\mathrm{TV} ] + u \right) \leq \exp\left\{ -\frac{n u^2}{8 \tau_{\min}} \right\}
\]
uniformly over $\bP \in \mathcal{P}$.

It remains to bound $\E[ \norm{ \hat{\pi}_n - \pi }_\mathrm{TV} ]$. As a corollary of (3.31) in \cite{paulin2015concentration}, we have
\[
\E[ \norm{ \hat{\pi}_n - \pi }_\mathrm{TV} ] \leq \sqrt{ \frac{1}{n \gamma / 2} } \sum_{\hat{x}} \sqrt{ \pi(\hat{x}) } \leq \sqrt{ \frac{m^2}{n \gamma} },
\]
where $\gamma$ is the pseudo-spectral gap and $m^2$ is the number of states of the block process. By uniform ergodicity, this bound is uniform in $\bP$.

Putting everything together, we obtain a uniform exponential bound
\begin{align*}
\sup_{\bP \in \mathcal{P}} \pp\left( \norm{ \hat{\pi}_n - \pi }_\mathrm{TV} \geq \epsilon \right)
&= \sup_{\bP \in \mathcal{P}} \pp\left( \norm{ \hat{\pi}_n - \pi }_\mathrm{TV} \geq \E[ \norm{ \hat{\pi}_n - \pi }_\mathrm{TV} ] + (\epsilon - \E[\cdot]) \right) \\
&\leq \exp\left\{ -\frac{ (\epsilon - \E[\cdot])^2 }{ 2 \tau_{\min} \cdot \frac{4}{n} } \right\} \\
&\coloneqq \alpha(n),
\end{align*}
where $\alpha(n)$ is independent of $\bP$ and $\alpha(n) \to 0$ as $n \to \infty$. Thus, $\hat{\pi}_n$ is a uniformly consistent estimator of $\pi$, and hence of $\bP$.

We now define the test function
\[
\varphi_n = \one\left\{ \norm{ \hat{\pi}_n - \pi_0 }_\mathrm{TV} \geq \epsilon/2 \right\},
\]
where $\pi_0$ is the block stationary distribution corresponding to $\bP_0$. Since the map $\bP \mapsto \pi$ is injective over $\mathcal{P}$, testing $\pi = \pi_0$ is equivalent to testing $\bP = \bP_0$. Therefore, this test is uniformly consistent.
\end{proof}

\begin{proof}[Proof of Condition (C4)]
Lemma~\ref{lem:lan_clt} establishes that the score function
\[
\Delta_n = \frac{1}{\sqrt{n}} \nabla_{\bP_0} \log \mathbb{L}(\{X_i\} \mid \bP)
\]
converges in distribution under $\bP_0$ to a normal distribution with mean zero and covariance matrix $\mathcal{I}(\bP_0)$. Lemma~\ref{lem:uniform_ll} establishes uniform convergence of the log-likelihood and observed information over a neighborhood of $\bP_0$. Together, these lemmas imply the expansion
\[
\log \mathcal{L}(\bx_n \mid \bP_0 + n^{-1/2} h_n) - \log \mathcal{L}(\bx_n \mid \bP_0)
= h_n^\T r_n - \tfrac{1}{2} h_n^\T \mathcal{I}(\bP_0) h_n + o_p(1),
\]
uniformly for $h_n \to h$ in compact sets, as required. See Theorem 7.2 of \cite{van2000asymptotic}.
\end{proof}

\begin{proof}[Proof of Condition (C5)]
Let $s_n = n^{-1/2} r_n$, where $r_n = \nabla_{\bP_0} \log \mathcal{L}(\bx_n \mid \bP)$ is the score function. Consider a Taylor expansion of $\mathbb{E}_\bP[s_n]$ around $\bP_0$:
\[
\mathbb{E}_\bP[s_n] = \mathbb{E}_{\bP_0}[s_n] + \nabla_{\bP} \mathbb{E}_\bP[s_n] \big|_{\bP = \bP_0} (\bP - \bP_0) + R_n(\bP),
\]
where $R_n(\bP)$ contains the higher-order terms. At $\bP = \bP_0$, the derivative $\nabla_{\bP} \mathbb{E}_\bP[s_n]$ equals the Fisher information matrix $\mathcal{I}(\bP_0)$, which is invertible by Lemma~\ref{lem:lan_clt}. Thus, for small $\| \bP - \bP_0 \|$, the linear term dominates, and there exist constants $n_0 \in \mathbb{N}$, $\epsilon > 0$, and $c > 0$ such that for all $n > n_0$ and $\| \bP - \bP_0 \| < \epsilon$,
\[
\left\| \mathbb{E}_\bP[s_n] - \mathbb{E}_{\bP_0}[s_n] \right\| \geq c \left\| \bP - \bP_0 \right\|.
\]
\end{proof}

\begin{proof}[Proof of Condition (C6)]
We apply the concentration inequality in Lemma~\ref{lem:concentration_inequalities} to the scaled score function
\[
s_n = f(x_0, \dots, x_n) = \frac{1}{\sqrt{n}} \nabla_{\bP_0} \log \mathcal{L}(\bx_n \mid \bP).
\]
The score $s_n$ is an average of terms $\log \bP(x_i \mid x_{i-1})$. Changing one $x_i$ affects at most two such terms, each contributing $O(1)$, so the total change in $s_n$ is $O(1/n)$. That is, there exists a constant $c > 0$ such that
\[
\| s_n(\bx) - s_n(\bx') \| \le \frac{c}{n}
\]
whenever $\bx$ and $\bx'$ differ at a single index, yielding $\sum c_i^2 \le c^2/n$. By Lemma~\ref{lem:mixing}, $\tau_{\min}$ is uniformly bounded, so applying the inequality gives
\[
\pp_\bP\left( \| s_n - \mathbb{E}_\bP[s_n] \| > u \right) \le 2 \exp\left\{ -\frac{u^2}{2 C/n} \right\}
\]
for some $C > 0$ independent of $\bP$.
\end{proof}

\begin{proof}[Proof of Condition (C7)]
Define the empirical average $\bar{g}_S = \frac{1}{S} \sum_{s=1}^{S} g_s$. By the corollary to Lemma~\ref{lem:mixing}, the block process $\{ \hat{X}_j^R \}$ is uniformly geometrically ergodic with mixing time upper bound $\bar{\tau}_{\min, \hat{X}}$. Since $g_s \in [0,1]$, the function $\bar{g}_S$ satisfies the bounded differences condition with $c_s = 1/S$. Applying Lemma~\ref{lem:concentration_inequalities} gives
\[
\pp\left( \bar{g}_S < \mathbb{E}[\bar{g}_S] - u \right) \le \exp\left\{ -\frac{u^2}{2 \bar{\tau}_{\min, \hat{X}} / S} \right\},
\]
and so (C7) is satisfied with constant $c = \bar{\tau}_{\min, \hat{X}}$.
\end{proof}

\subsection{Proofs of Theorems}

\begin{proof}[Proof of Theorem~\ref{the:bvm}]
  The result follows from the main argument of \cite{connault2014weakly}, which establishes a Bernstein--von Mises theorem under weak dependence, provided conditions (C1)--(C7) are satisfied. Verification of these conditions under Assumptions (A1) and (A2) is provided in the preceding section.
\end{proof}
  
\begin{proof}[Proof of Theorem~\ref{the:spectral_consistency}]
From Theorem~1, the Bernstein--von Mises posterior approximation implies
\begin{equation*}
  \hat{\bP} = \bP_0 + n^{-1/2}  \bQ + o_p(n^{-1/2}),
\end{equation*}
where $\bQ \sim \mathcal{N}(0, 2 \mathcal{I}^{-1}(\bP_0))$ and $\bP_0 = \exp\{\bL_0 \Delta\}$. Let the spectral decomposition of $\bP_0$ be
\begin{equation*}
  \bP_0 = \sum_{k=1}^m \Lambda_k \phi_k^0 (\psi_k^0)^\T.
\end{equation*}

Let $\hat{\bP} = \bP_0 + \bP_\delta$ denote a small perturbation of $\bP_0$. By standard first-order matrix perturbation theory \citep[see][]{stewart1990matrix}, the eigenvalues and eigenvectors of $\bP_0$ are perturbed as
\begin{align*}
  \hat{\Lambda}_k &= \Lambda_k^0 + (\psi_k^{0})^\T \bP_\delta \phi_k^0 + o\left(\norm{\bP_\delta}_\lambda\right), \\
  \hat{\phi}_k &= \phi_k^0 + \bR_k \bP_\delta \phi_k^0 + o\left(\norm{\bP_\delta}_\lambda\right), \\
  \hat{\psi}_k &= \psi_k^0 + \bR_k^\T \bP_\delta^\T \psi_k^0  + o\left(\norm{\bP_\delta}_\lambda\right),
\end{align*}
where $\bR_k$ is the resolvent matrix as defined in the theorem. Substituting $\bP_\delta = n^{-1/2} \bQ$ and noting that $o\left(\norm{\bP_\delta}_\lambda\right) = o(n^{-1/2})$ gives the desired convergence for the eigenvectors,
\begin{equation*}
  \sqrt{n} (\hat{\phi}_k - \phi_k^0) \rightarrow \bR_k \bQ \phi_k^0,
  \quad
  \sqrt{n} (\hat{\psi}_k - \psi_k^0) \rightarrow \psi_k^0 \bQ \bR_k.
\end{equation*}

For the eigenvalues, define $\delta \Lambda_k = n^{-1/2} \psi_k^0 \bQ \phi_k^0$, then 
\begin{equation*}
  \hat{\Lambda}_k = \Lambda_k^0 + \delta \Lambda_k + o(n^{-1/2}).
\end{equation*}
We know $\hat{\lambda}_k = \Delta^{-1} \log\{\hat{\Lambda}_k\}$ so
\begin{equation*}
  \hat{\lambda}_k = \Delta^{-1} \log\{\Lambda_k^0 + \delta \Lambda_k\} 
  = \Delta^{-1}\left[\log\{\Lambda_k^0\} + \frac{\delta \Lambda_k}{\Lambda_k^0} + o(\delta \Lambda_k)\right]
\end{equation*}
where the second equality is the first-order Taylor expansion about $\Lambda_k^0$. Thus,
\begin{align*}
  \hat{\lambda}_k &= \lambda_k^0 + \Delta^{-1}\left[\frac{\delta \Lambda_k}{\Lambda_k^0} + o(\delta \Lambda_k)\right], \\
  \hat{\lambda}_k - \lambda_k^0 &= \Delta^{-1}\left[\frac{n^{-1/2} \psi_k^0 \bQ \phi_k^0}{\exp\{\lambda_k^0 \Delta\}} + o(\delta \Lambda_k)\right], \\
  n^{-1/2} (\hat{\lambda}_k - \lambda_k^0) &\rightarrow \Delta^{-1} \exp\{-\lambda_k^0 \Delta\} \psi_k^0 \bQ \phi_k^0,
\end{align*}
completing the proof.
\end{proof}

\end{document}